\documentclass[12pt,journal,cspaper,compsoc]{IEEEtran}
\usepackage{algorithm}
\usepackage{algorithmic}
 \usepackage{fullpage}
\usepackage{graphicx}
\usepackage{amsthm}
\usepackage{amsmath}
\usepackage{subfig}
\theoremstyle{definition}

\theoremstyle{remark}

\theoremstyle{plain}

\newcommand{\algorithmroutine}[1]{-- \textit{#1}}

\begin{document}
%
\title{A novel Quorum Protocol}

\author{Parul Pandey, Maheshwari Tripathi \thanks{Maheshwari Tripathi is working with the Computer Science Dept of IET,India} }

%

\IEEEcompsoctitleabstractindextext{%
\begin{abstract}
 One of the traditional mechanisms used in distributed systems for maintaining the consistency of replicated data is voting.
  A problem involved in voting mechanisms is the size of the Quorums needed on each access to the data. In this paper, we present a novel and efficient distributed algorithm for managing replicated data. We impose a logical wheel structure on the set of copies of an object. The protocol ensures minimum read quorum size of one, by reading one copy of an object while guaranteeing fault-tolerance of write operations.Wheel structure has a wider application area as it can be imposed in a network with any number of nodes.
\end{abstract}

}
\maketitle
\bibliographystyle{plain}
\IEEEdisplaynotcompsoctitleabstractindextext
\IEEEpeerreviewmaketitle
\section{Introduction}
%
%

%
%
%
%
\IEEEPARstart{I}{n} a distributed database system, data is replicated \cite{understreplicatn}, \cite{optimistic}, \cite{citeulike:985059} to achieve fault-tolerance. One of the most important advantages of replication is that it masks and tolerates failures in the network gracefully and increases availability.  In particular, the system remains operational and available to the users despite failures.
 In case of multiple access a problem that must be solved while using replication is how to maintain the copies in a consistent state \cite{dangers}. To keep logical data consistent, there must exist a control protocol responsible for synchronizing the access.  A popular method for maintaining consistency of replicated data is weighted voting  \cite{weightedvoting} which is a generalization of the majority consensus method presented in \cite{majority}. In the quorum consensus (QC) \cite{qc}, \cite{scienceorfiction} algorithm, we assign a non-negative weight \cite {assign} to each copy $x_{A}$ of $x$. We then define a read threshold $RT$ and write threshold $WT$ for $x$, such that both $2WT$ and $(RT + WT)$ are greater than the total weight of all copies of $x$. A read (or write) quorum of $x$ is any set of copies of $x$ with a weight of at least $RT$ (or $WT$).  For better performance, some logical structure is imposed on the network, and the quorums are chosen under the consideration of such structures. Such logical structures include the tree  \cite{tree90}, diamond \cite{diamond}, ring \cite{ring}, triangular mesh \cite{mesh}, and grid \cite{grid} structures. A geometric approach for dealing with logical structures is proposed in  \cite{coterie}.

In this paper we propose a novel protocol, which is called The \emph{Wheel Quorum Consensus Protocol} or simply \emph{The Wheel Protocol}, for managing replicated data. In this protocol, the sites in the network are logically organized into a wheel structure. This protocol can be viewed as specialized version of ring and tree protocol. This protocol has an upper hand on both tree and ring protocol, unlike tree and ring protocol it's read quorum size never exceeds one , which is minimum among all.
 As compared to tree, grid, diamond and mesh protocol, wheel protocol is very flexible in arranging nodes in a network into the logical structure. Any number of nodes can be easily organized into a wheel structure.

The paper is organized as follows. In Section 2 we describe the system model. Section 3 discusses wheel quorum protocols which elaborates the motivation behind it, wheel structure and its quorum construction for read and write.
\section{Model}
A distributed system consists of a set of distinct sites that communicate with each other by sending messages over a communication network. No assumptions are made regarding the speed, connectivity, or reliability of the network. It is assumed that sites are fail-stop \cite{failstop} and communication links may fail to deliver messages.

Replication of data is achieved by storing copies of the same logical data item at different nodes. Read and write operations can be performed on replicated data. A node needs to obtain permission from a number of copies (quorum) before performing the operation using a control protocol.

In a replicated database, copies of an object may be stored at several sites in the network. Multiple copies of an object must appear as a single logical object to the transaction. This is termed as one-copy equivalence \cite{goodman87} and is enforced by the replica control protocol. The correctness criteria for replicated databases is one-copy serializability \cite{goodman87}, which ensures one-copy equivalence and serializable execution of transactions. In order to ensure one-copy equivalence, a replicated object \emph{z} may be read by reading a read quorum of copies, and it may be written by writing a write quorum of copies. The following restriction is placed on the choice of quorum assignments:

\textbf{Quorum Intersection Property: }For any two operations o[Z] and $\acute{o}$[z] on an data item  $x$, where at least one of them is a write, the quorums must have a nonempty intersection.

Version numbers or timestamps are used to identify the current copy in a quorum.
Each node is logically characterized by few attributes as shown in figure 1.
\textbf{ID} which is a unique sequential ID. In our discussion, IDs are numbered as 0, 1, 2, 3,... $n$.
\textbf{Node \_ Location} is the location where the node is physically residing. In other words this is the address of a node in the network.
\textbf{HUB} contains the ID of the node in the wheel which is currently acting as hub. In our discussion, ID of the HUB node is 0.
\textbf{SUC} contains the ID of the successor $w_{i+1}$, which is the next node in the wheel.
\textbf{PRED} contains the ID of the predecessor $w_{i-1}$, which is the previous node in the wheel.

\begin{figure}[h]
\begin{center}
\includegraphics[width=2in]{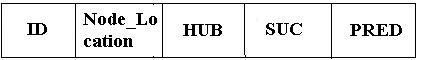}
\label{Wheel node}
\end{center}
\caption{Wheel Structure}
\end{figure}

The election quorum ensures that the HUB's ID is always 0.
\section{Wheel Quorum Protocol}

\subsection{Motivation}
Tradeoff between the cost for reading, writing, data availability and node fault tolerance is the deciding feature of all existing control protocols for replicated data . For example, the read-one write-all scheme needs only one copy as read quorum, but has the convenience of having a write quorum equal to the total number of copies ( thus not tolerating a single node of failure).

The main motivation for our work was to develop a protocol which had a constant minimum cost for reading, while maintaining an acceptable cost for writing, since we are interested in systems where read operations are much more frequent than write operations.

To achieve this property, a logical wheel structure will be imposed on the set of copies of the object . This structure is used by operations to determine the copies that must be read or written. Figure 2, represents 4 nodes arranged in a wheel structure. Wheel logical structure can be arranged on any number of nodes, whereas other logical structures have constraints with nodes arrangement. We note that this structure is logical, and does not have to correspond to the actual physical structure of the network connecting the sites, storing the copies. This wheel structure is used to motivate the protocol.
\subsection{The Wheel Structure}
Let $W_{n}$ = {$w_{0}$, $w_{1}$, $w_{2}$,...,$w_{n-1}$} be the set of nodes that store  copies of a replicated data item. \textbf{A wheel}, $W_{n}$ is a logical structure with \emph{n} nodes, formed by connecting a single node called HUB to all vertices of an \emph{(n-1)} cycle. The numerical notation for wheels is used inconsistently in the literature: some authors instead use \emph{n} to refer to the length of the cycle, so their $W_{n}$ is the graph we would denote as $W_{n+1}$. All nodes in the cycle maintain adjacency relationship by maintaining ID's of their successor and predecessor. Each node is defined by attributes ID, Node\_Location, HUB, Suc, and Pred as shown in figure 1.
Wheel structure is easily imposed on the set of nodes by selecting first node as HUB and adding other nodes as spokes in cycle by defining the successor (Suc(i)) , predecessor (Pred(i))  operations and by setting HUB in each spoke.
Other operations are GetPermission(i) and rand(1..n).

 GetPermisson(i), returns TRUE if the node $w_{i}$ allows access to its own copy of the item. GetPermisson(i) returns FALSE when either node $w_{i}$ refuses access or cannot be contacted due to failure. rand(1..n) selects and returns random number from 1 to n, where n is the number of nodes in wheel. This random number represents ID of selected node.

\begin{figure}[h]
\begin{center}
\includegraphics[width=2.5in]{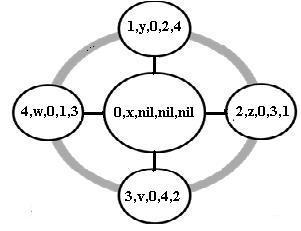}
\label{Wheel Sructure}
\end{center}
\caption{Wheel Structure}
\end{figure}
One of the restrictions imposed by the suggested implementation for collecting read quorums is that the reads are directed to a specific copy: the HUB. This has the advantage that if the HUB is up, read  operations accesses a single copy. Read locality may, however, be sacrificed and the HUB may become a bottleneck. To solve this problem, it is desirable to gather a quorum of several relatively-local copies rather than one very remote HUB copy. This approach could also be used for organizing the wheel structure of the copies. For example, consider a network composed of two relatively distant segments: the HUB could be placed in one of the segment and the other nodes of the wheel in the other segment. In such an organization, transactions executing in a particular segment will use the quorum which is less expensive. If the HUB is in the transaction's network segment, the HUB will be accessed. Otherwise, the transaction will access any two adjacent nodes of the wheel. The functions depicting Read and Write quorum should be appropriately modified to enforce this policy. Whereas, one policy of election quorum is already suggested in this paper to avoid the problem of HUB bottleneck.
\subsection{The Wheel Protocol}
In this protocol, all copies of a replicated data item are organized into a wheel structure. Specific algorithms are used for read and write quorums construction. There is one election algorithm for electing new HUB in case of failure of HUB or in case load threshold exceeds its limit. These algorithms use the adjacency information to guarantee quorum intersection, and to maintain the quorum sizes small. There are three type of quorums, Read, Write, and Election quorum.

\textbf{Read Quorum } is formed by getting access permission from HUB.

\textbf{Write quorum }is obtained by getting access permission from HUB and half of alternating nodes in the cycle, thus requiring the majority of the total number of copies.
As an example, consider a replicated data item with six copies arranged in a wheel structure as shown in figure 3. Eligible read quorum is {0} ( i.e. {HUB}) and sets eligible for write quorum are : \{0,1,3,5\}, \{0,1,2,4\}, \{0,3,5,2\}, \{0,4,1,3\} and \{0,5,2,4\}.

Notice that eligible quorums are coteries, satisfying the minimality and intersection properties\footnote{The fact that the quorums are distinct and have the same size shows that they satisfy the minimality property: the intersection property will be shown later, when providing the protocol correctness}.

\begin{figure}[h]
\begin{center}
\includegraphics[width=1.5in]{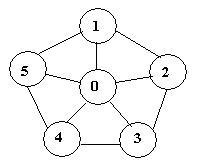}
\label{Wheel with 6 nodes}
\end{center}
\caption{6 copies organized into a logical structure}
\end{figure}

Election quorum is called in two situations
\begin{enumerate}
\item When HUB crosses its load threshold
\item When HUB is unavailable
\end{enumerate}
In both the above cases, the node initiating election quorum algorithm, selects randomly any 2 adjacent nodes, checks their version and makes the latest one the HUB by changing the location address between the old HUB and the newly elected one. This logically swaps the location of the two nodes. Other nodes are unaffected as they identify HUB by its ID, which is 0. Only the node\_location is changed.

Advantages of this Election Quorum are-
\begin{enumerate}
\item HUB is never overloaded, as it gets swapped with a latest node whenever load \_ threshold crosses its limit.
\item Improved load distribution. Assuming that each node in cycle has the equal probability of being selected as a new HUB, no node will be working as HUB for a longer time.
\item Constant minimum possible Read Quorum size of one. As, even if HUB is failed , it will be replaced with a new HUB. Thus ensuring that a request always reads data from HUB.

\end{enumerate}
Without using election quorum, in the failure of HUB, Read Quorum can be achieved by accessing any 2 adjacent nodes in the cycle, which is double the cost of doing it with HUB. Our system has more number of reads as compared to write , so reads will keep on costing double till HUB recovers. All this can be avoided by using election quorum and electing new HUB. This way, present as well as subsequent reads can be satisfied by reading only HUB.

\subsubsection{Quorum Construction}
There are three algorithms for the wheel protocol. Algorithm 1, 2, 3 for read, write and election quorum respectively.

Algorithm 1 defines read quorum construction. This algorithm returns the HUB as the read quorum. In case of a HUB failure, the new HUB is elected by invoking the ElectionQuorum Protocol, which uses a random node in cycle.

\begin{algorithm}
    \caption{Read Quorum(i)}
    \label{Read Quorum}
    \begin{algorithmic}
        \IF {Empty(Wheel)}
        \STATE Return(nil)
        \ELSIF {GetPermission(HUB) is False}
                \STATE r= rand(1 .. n)
                \STATE Get ElectionQuorum(r)
                 \STATE  Return(HUB)
                 \ELSE
                 \STATE Return(HUB)
        \ENDIF
  \end{algorithmic}
\end{algorithm}

Algorithm 2 is to find write quorum. This protocol collects majority of nodes forming quorum between nodes in cycle of wheel in list, Quorum\_list[]. This Quorum\_list[] along with HUB makes write quorum. Protocol tries to form write quorum with current\_node by traversing the cycle until, either a quorum is obtained or all copies have been examined (in which case quorum was not obtained and the request for writing is refused). In case of HUB failure Election Quorum elects a new HUB.

\begin{algorithm}
    \caption{Write Quorum(i)}
    \label{Write Quorum}
     \algorithmroutine{Main routine}
    \begin{algorithmic}[1]
        \STATE nodes\_covererd=0
        \STATE current\_node = i
         \IF {GetPermission(HUB) is False}
          \STATE n= random(cycle nodes)
          \STATE Get ElectionQuorum(n)
          \STATE GetPermission(HUB)
         \ENDIF

         \IF {current\_node is HUB}
            \STATE current\_node = rand(1..n)
         \ENDIF
         \WHILE{{Empty QuorumList[] \AND $nodes\_covered < n$}}
         \STATE Quorum\_list[]= Check(current\_node)
        \STATE current\_node=Suc(current\_node)
        \STATE nodes\_covered++
         \ENDWHILE
    \STATE Return(HUB $\bigcup$ QuorumList[])
\end{algorithmic}
 \algorithmroutine{Check($i$)}
  \begin{algorithmic}[1]
  \STATE Quorum\_list[] = null
  \STATE fail= nodes\_checked=0
  \WHILE{$Fail\neq 1$ \AND $nodes\_checked < \lfloor n/2\rfloor$}
   \IF {GetPermission(i)}
   \STATE Quorum\_list.add(i)
   \STATE i=Suc(Suc(i))
   \STATE nodes\_checked++
   \ELSE
   \STATE Fail=1
   \ENDIF
   \ENDWHILE
  \IF {Fail}
    \STATE Quorum\_list.flushall()
    \STATE return(Quorum\_list[])
  \ELSE
    \STATE return(Quorum\_list[])
  \ENDIF
  \end{algorithmic}
\end{algorithm}

In case of HUB failure, Election Quorum (Algorithm 3) elects a new HUB. This protocol can be called in two conditions. First when the HUB has failed or second whenever HUB exceeds it's load threshold. Election quorum selects two adjacent nodes(using successor function), selects the node with latest value and makes it the HUB.
\begin{algorithm}
    \caption{Election Quorum(i)}
    \label{election Quorum}
    \begin{algorithmic}[1]
    \STATE current\_node=i
    \STATE Quorum=0
    \STATE nodes\_done=0
    \IF {current\_node is HUB}
    \STATE current\_node=rand(1..n-1)
    \ENDIF
    \WHILE{{Quorum is Empty \OR $nodes\_done < n$}}
      \IF {current\_node is accessible}
        \IF {SUC(current\_node) is accessible}
          \STATE Latest\_node=Node\_Location with most recent value
          \STATE Swap Node\_Location of HUB and Latest\_node
          \STATE Quorum=Latest\_node
        \ELSE
          \STATE current\_node=Suc(Suc(current\_node))
          \STATE nodes\_done=nodes\_done + 2
        \ENDIF
      \ELSE  		
        \STATE current\_node=Suc(current\_node)
		\STATE nodes\_done=nodes\_done+1
      \ENDIF
   \ENDWHILE
\end{algorithmic}
\end{algorithm}

\subsubsection{Quorum Size}
An outstanding feature of Wheel quorum is its minimum read quorum size which is always one. This is achieved by reading only HUB. HUB is always made available even if existing one is failed by election quorum.  This is minimum read quorum size achieved by any algorithm.

Write quorum size is \textbf{\emph{}$\lceil(n-1)/2\rceil+1$}, including the HUB.
\subsubsection{Proof of Correctness and Non-equivalence with Vote Assignment}
To show protocol correctness, it must be shown that no two conflicting operations are permitted to occur at the same time. Following theorems prove correctness of our algorithms.

\newtheorem{thm}{Theorem}[subsection]
\begin{thm}
In a wheel of size n, the Wheel Protocol guarantees a non-empty intersection between any read and write quorums.
\end{thm}
\begin{proof}
The proof follows from read and write quorum construction algorithms. The read quorum is formed by only HUB in the wheel and write quorum selects HUB and alternate nodes from (n-1) nodes of cycle. Both of them will definitely contain HUB and thus ensures non-empty intersection between any read and write quorums.
\end{proof}

\begin{thm}
In a wheel of size n, the Wheel Protocol guarantees that there is a non empty intersection between any write quorums.
\end{thm}
\begin{proof}
 It follows from the fact that write quorum is formed by majority of copies \textbf{\emph{}$\lceil(n-1)/2\rceil+1$} which includes HUB in the wheel. Since, each write quorum must include HUB, it is guaranteed that the intersection between any write quorums is non-empty.
\end{proof}
An interesting property of wheel protocol is that the coterie it generates cannot be generated by any vote assignment in the voting protocol\cite{weightedvoting}.
\begin{thm}
There is no vote assignment equivalent to the wheel protocol.
\end{thm}
\begin{proof}
By contradiction. Consider the wheel in figure 2.
Let $v_{0}$, $v_{1}$,..., $v_{5}$ be the vote assigned to the six copies and $V_{i}$ be the total number of votes. Consider, the two eligible write quorum sets of copies \{0, 1, 3, 4\} and \{0, 2, 4, 5\}, two other sets that are not eligible quorums\{0, 2, 3, 4\} and \{0, 1 ,4 , 5\}. For a vote assignment to be equivalent to the wheel protocol the following must hold

\begin{align*}
v_{0}+v_{1}+v_{3}+v_{4} > \frac{V_{i}}{2}\hspace{1cm}          (1)\\
v_{0}+v_{2}+v_{4}+v_{5} > \frac{V_{i}}{2}\hspace{1cm}              (2)\\
v_{0}+v_{2}+v_{3}+v_{4} < \frac{V_{i}}{2} \hspace{1cm}              (3)\\
v_{0}+v_{1}+v_{4}+v_{5} < \frac{V_{i}}{2} \hspace{1cm}              (4)
\end{align*}

For (1) , (2) there is a quorum and (3), (4) there is no quorum. Solving (1) and (3) we conclude that  $v_{1}$$>$$v_{2}$.
Solving (2) and (4) we conclude that  $v_{2}$$>$$v_{1}$, which is a contradiction. Therefore, there is no vote assignment (using positive integers) that satisfies both condition, and the theorem follows.

\end{proof}
Different logical structures have been expoited in \cite{wheel} and message overhead analysis of wheel is done in \cite{msgovrhd}
\bibliography{parulpandey}

\begin{thebibliography}{10}

\bibitem{goodman87}
P.~A. Bernstein and N.~Goodman.
\newblock A proof technique for concurrency control and recovery algorithms for
  replicated databases.
\newblock {\em Distributed Computing, Springer- Verlag}, 2( 1):32-44, January
  1987.

\bibitem{mesh}
Yao-Jen Chang.
\newblock A triangular-mesh-based approach to fault-tolerant distributed mutual
  exclusion.
\newblock Master's thesis, National Sun Yat-sen University, June, 1995.

\bibitem{tree90}
Amr E.~Abbadi Divyakant~Agrawal.
\newblock The tree quorum protocol: An efficient approach for managing
  replicated datain.
\newblock {\em Proceedings of the 16th International Conference on Very Large
  Data Bases (1990),}, pages pp. 243--254., 90:.

\bibitem{diamond}
Ada Wai-Chee Fu, Yat~Sheung Wong, and Man~Hon Wong.
\newblock Diamond quorum consensus for high capacity and efficiency in a
  replicated database system.
\newblock {\em Distrib. Parallel Databases}, 8:471--492, October 2000.

\bibitem{assign}
Hector Garcia-Molina and Daniel Barbara.
\newblock How to assign votes in a distributed system.
\newblock {\em J. ACM}, 32:841--860, October 1985.

\bibitem{weightedvoting}
H.~Gifford.
\newblock Weighted voting for replicated data.
\newblock {\em in Proceedings of 7th Symposium on operating Systems, ,ACM},
  pages pp 150--162, 1979.

\bibitem{dangers}
Jim Gray, Pat Helland, Patrick O'Neil, and Dennis Shasha.
\newblock {The dangers of replication and a solution}.
\newblock In {\em SIGMOD '96: Proceedings of the 1996 ACM SIGMOD international
  conference on Management of data}, volume~25, pages 173--182, New York, NY,
  USA, June 1996. ACM.

\bibitem{citeulike:985059}
Yi~Lin, Bettina Kemme, Marta Pati\~{a}±o Mart\~{a}­nez, and Ricardo
  Jim\~{a}{\copyright}nez-Peris.
\newblock {\em {Consistent Data Replication: Is It Feasible in WANs?}}
\newblock 2005.

\bibitem{qc}
M.~L. Liu, D.~Agrawal, and El~A. Abbadi.
\newblock {Abbadi. On the implementation of the quorum concensus protocol}.
\newblock In {\em In Proc. Parallel and Distributed Computing Systems}, 1995.

\bibitem{ring}
Nabor~C. Mendonça and Ricardo~O. Anido.
\newblock The hierarchical ring protocol: {An} efficient scheme for reading
  replicated data.
\newblock Technical Report DCC-93-02, Department of Computer Science,
  University of Campinas, February 1993.
\newblock In English, 30 pages.

\bibitem{wheel}
M.Tripathi Parul~Pandey.
\newblock Exploiting logical structures to reduce quorum sizes of replicated
  databases.
\newblock {\em Advanced Computing : An Interntional Journal (2012),},
  3:99--104, January 2012.

\bibitem{msgovrhd}
M.~Tripathi P.Pandey.
\newblock Message overhead analysis of quorum protocol.
\newblock {\em in Proceedings of International Conference on Advances in
  Computing}, 174:pp 237--245, 2012.

\bibitem{grid}
M.~H.~Ammar S.~Y.~Cheung and M.~Ahamad.
\newblock The grid protocol: A high performance scheme for maintaining
  replicated data.
\newblock {\em IEEE Transactions on Knowledge and Data Engineering}, Vol. 4,
  No. 6:pp. 582--592, Dec. 1992.

\bibitem{optimistic}
Yasushi Saito and Marc Shapiro.
\newblock {Optimistic replication}.
\newblock {\em ACM Comput. Surv.}, 37(1):42--81, March 2005.

\bibitem{failstop}
Richard~D. Schlichting and Fred~B. Schneider.
\newblock {Fail-Stop Processors: An Approach to Designing Fault-Tolerant
  Computing Systems}.
\newblock {\em Computer Systems}, 1(3):222--238, 1983.

\bibitem{majority}
Robert~H. Thomas.
\newblock {A Majority consensus approach to concurrency control for multiple
  copy databases}.
\newblock {\em ACM Trans. Database Syst.}, 4(2):180--209, June 1979.

\bibitem{understreplicatn}
M.~Wiesmann, F.~Pedone, A.~Schiper, B.~Kemme, and G.~Alonso.
\newblock Understanding replication in databases and distributed systems.
\newblock In {\em In Proceedings of 20th International Conference on
  Distributed Computing Systems (ICDCS'2000}, pages 264--274, 2000.

\bibitem{scienceorfiction}
Avishai Wool.
\newblock {Quorum systems in replicated databases: science or fiction}.
\newblock {\em Bull. IEEE Technical Committee on Data Engineering}, 21:3--11,
  1998.

\bibitem{coterie}
Y.C.Kuo and S.T. Huang.
\newblock A geometric approach for constructing coteries and k-coteries.
\newblock {\em IEEE Transaction Parallel and Distributed Systems}, 8(4):402
  411, April 1997.

\end{thebibliography}
\end{document}